\documentclass[a4paper, 12pt]{article}

\usepackage{amssymb, amsmath, amsthm}
\usepackage[round]{natbib}
\usepackage{graphicx}
\RequirePackage[colorlinks,allcolors=blue]{hyperref}

\newtheorem{proposition}{Proposition}

\newcommand{\E}{\mathop{\mathbb E}\nolimits}
\newcommand{\R}{\mathbb{R}}
\newcommand{\N}{\mathbb{N}}
\newcommand{\I}{\mathbb{I}}
\newcommand{\pto}{\ensuremath{\stackrel{p}{\to}}}

\newcommand{\pneq}{\ensuremath{\stackrel{\text{a.\,s.}}{=}}}

\renewcommand{\P}{\mathbb{P}}

\begin{document}
	
	\begin{center}
		\large \bf
		ESTIMATING THE DISTRIBUTION OF DISPLACEMENTS
	\end{center}
	
	\vspace{0.5cm}
	
	\centerline{A.\ S.\ Kurennoy\footnote{%
			 Email: {\tt akurennoy@cs.msu.ru}\\
			This research is supported by the Russian Foundation for Basic Research Grant 17-01-00125 and the Russian Science Foundation Grant 15-11-10021.}}
	
	\vspace{0.7cm}
	
	\centerline{\today}
	
	\vspace{0.5cm}
	
	\begin{center}
		\bf{ABSTRACT}
	\end{center}
	
	\noindent
	In this paper, we propose a method for estimating the distribution of time differences between connected events (such as ad impressions and corresponding customer calls). A special feature of this method is that it does not require matching those connected events with each other. The method is very simple to use as it essentially consists of computing an ordinary least squares estimator.
	
	\vspace{0.1cm}
	
	{\small
		
		}	
	
	\section{Introduction}
	
	The term \emph{displacement} in this paper refers to the time difference between two connected events.
	Such events can be, for instance, a visit of a customer to an online store and a call made by that customer to the store. The call is made after some delay, i.\,e. it is displaced in time with respect to the visit. In this example, displacement is the same as \emph{time to conversion}.
	
	One may be interested in the distribution of displacements for various reasons.
	For example, the knowledge of the distribution of delays between visits and calls can be used for managing the pool of telephone numbers more optimally.
	
	Estimating the distribution of displacements is straightforward if they are directly observed. In the aforementioned example, it is the case if calls can be accurately matched with the corresponding visits. However, in practice, such matching is not always possible.
	
	In this paper, we propose a method for estimating the distribution of displacements that relies solely on the counts of events before and after displacing and does not require observing displacements. In order to apply this method in the example above one only needs to know the number of visits and calls in consecutive time intervals and does not have to match visits with calls at all.
	
	
	The proposed method 
	is described in Section~\ref{sec:method}. Section~\ref{sec:mc} presents some simulation results. All formal derivations are in the appendix.
	
	\section{The Proposed Method} \label{sec:method}
	
	Let the real line be divided into intervals of length $\ell$. Furthermore, let $n_t$ and $k_t$, $t \in\mathbb{Z}$, be the number of events in the $t$-th interval before and after displacing, respectively. In terms of the example mentioned in the introduction $n_t$ is the number of visits and $k_t$ is the number of calls. We assume that during displacing each event can be dropped with probability $1-q$, because, obviously, not all visits are converted into calls. As shown in the appendix (Proposition~\ref{prop:E}), given that displacements are distributed in the interval $[0,\,m\ell]$ for some nonnegative integer $m$, the expectation of $k_t$ conditional on $n_t$, $n_{t-1}$,\,\ldots,\,$n_{t-m}$ equals
	\begin{equation} \label{E=}
	\E[k_t\mid n_t,\,n_{t-1},\,\ldots,\,n_{t-m}] = \sum_{j=0}^m qp_j n_{t-j},
	\end{equation}
	where $p_j$ is the probability of an event to be displaced to the $j$-th interval ahead. If the times of events in each interval are distributed uniformly
	then 
	\begin{equation} \label{p=}
	p_0 = I_0,\quad p_j = I_j - I_{j-1},\quad j = 1,\,\ldots,\,m,
	\end{equation}
	where $I_j = \frac 1 \ell\int _{j\ell}^{(j+1)\ell} F(\tau)d\tau$, $j=0,\,1,\,\ldots,\,m$,
	with $F$ being the cumulative distribution function of displacements. Then, having estimates $\hat p_j$ of the probabilities $p_j$, $ j=0,\,1,\,\ldots,\,m$, we can estimate the integrals $I_j$, $j=0,\,1,\,\ldots,\,m$, as
	\begin{equation} \label{hI=}
	\hat I_j = \sum_{i=0}^j \hat p_i,\quad j=0,\,1,\,\ldots,\,m,
	\end{equation}
	and approximate $F$ by the following step function
	\begin{equation} \label{hF=}
	\hat F(\tau) = \begin{cases}
	0, & \tau < 0,\\
	\hat I_j, & j\ell \le \tau < (j+1)\ell,\quad j = 0,\,1,\,\ldots,\,m-1,\\
	1, & \tau \ge 1.
	\end{cases}
	\end{equation}
	We mention that the times of events in an interval are distributed uniformly if, for example, the events arrive according to a Poisson point process which intensity function is constant on the interval (see, e.g., \cite[Section~2.3]{Str:PPP:10}).
	
	Suppose that the variables $n_t$ and $k_t$ are observed for $t=1,\,\ldots,\,T$. In order to estimate $p_j$, $j=0,\,1,\,\ldots,\,m$, one can first find the ordinary least squeares estimates $\hat\beta_j$, $j=0,\,1,\,\ldots,\,m$, of the regression coefficients in \eqref{E=} and then compute
	\begin{equation} \label{hpj=}
		\hat p_j = \frac{\max\{\hat \beta_j,\, 0\}}{\sum_{i=0}^m \max\{\hat \beta_i,\, 0\}},\quad j=0,\,1,\,\ldots,\,m.
	\end{equation}
	If
	$$
		\hat \beta_j \pto qp_j\quad \forall\,j=0,\,1,\,\ldots,\,m
	$$ 
	then, by a standard continuity argument,
	$$
		\hat p_j \pto p_j\quad \forall\,j=0,\,1,\,\ldots,\,m.
	$$
	Conditions ensuring the consistency of ordinary least squares estimates can be found, for example, in \cite{Whi:ATE:01}.
	
	An alternative way of estimating the probabilities $p_j$, $j=0,\,1,\,\ldots,\,m$, would consist of computing an estimate $\hat q$ of the conversion rate (e.g. $\hat q = {\sum_{t=1}^T k_t}/{\sum_{t=1}^T n_t}$), then solving the following constrained optimization problem 
	\begin{equation*}
		\begin{array}{ll}
		\min_{b_0,\,b_1,\,\ldots,\,b_m\in\R} & \sum_{t=m+1}^T\left(k_t - \sum_{j=0}^m b_j n_{t-j}\right)^2\\
		\textup{subject to} & \sum_{j=0}^m b_j = \hat q\\
		& b_j \ge 0,\quad j=0,\,1,\,\ldots,\,m,
		\end{array}
	\end{equation*}
	and finally deviding the solution by $\hat q$. However, the estimates obtained in this way had inferior performance compared to the ones defined by \eqref{hpj=} in our simulation experiment.

	\section{Numerical Experiment} \label{sec:mc}

	In this section we illustrate our findings with a simulation experiment.
	
	In each replication we generated events for $N$ days according to an inhomogeneous Poisson process. Time was measured in minutes. The intensity function of the process was constant within each hour of the day, namely,
	$$
	\lambda(\tau) = \lambda_{([\tau/60]\!\!\!\!\mod 24)},\quad \tau \ge 0,
	$$
	with $\lambda_0=50$,  $\lambda_1=\lambda_{23}=63$,  $\lambda_2=\lambda_{22} =75$,  $\lambda_3=\lambda_{21}=88$, $\lambda_4=\lambda_{20}=100$, $\lambda_5=\lambda_{19}=110$, $\lambda_6=\lambda_{18}=120$, $\lambda_7=\lambda_{17}=129$, $\lambda_8=\lambda_{16}=136$, $\lambda_9=\lambda_{15}=142$, $\lambda_{10}=\lambda_{14}=146$, $\lambda_{11}=\lambda_{13}=149$, $\lambda_{12}=150$. The displacements were drawn from the uniform distribution on $[0,\,60]$. The parameter $m$ took values 1, 2, 3, 4, 5, 6, 10, 12, 15, 20, 30, 60 and in each case the parameter $\ell$ was set to $60 / m$. The number of days $N$ equaled 5, 10, 30, and 60. Finally, $q$ (the conversion rate) was equal to $0.01$, $0.05$, and $0.1$.
	
	For each combination of parameters we performed 500 replications. The results are presented in Figure~\ref{fig:1}. The upper line shows the average $L_2$ distance between the true distribution function of displacements and its estimate (for different values of $m$). The probabilities $p_j$ were estimated according to \eqref{hpj=}. The lower line depicts the minimum distance that can be attained by an approximation that is constant on the intervals $[j\ell,\,(j+1)\ell)$, $j=0,\,1,\,\ldots,\,m-1$, (this is the distance that would be observed if the probabilities $p_j$ were estimated perfectly).
	
	Clearly, the approximation accuracy is higher for larger values of $N$ and $q$ (i.\,e. when there is more data or when the conversion rate is higher). 
	
	\begin{figure}
		\caption{$L_2$ distance from the true distribution function for different values of $m$: minimum possible (lower line) and averaged over 500 simulations (upper line).} \label{fig:1}
		\includegraphics[width=\textwidth]{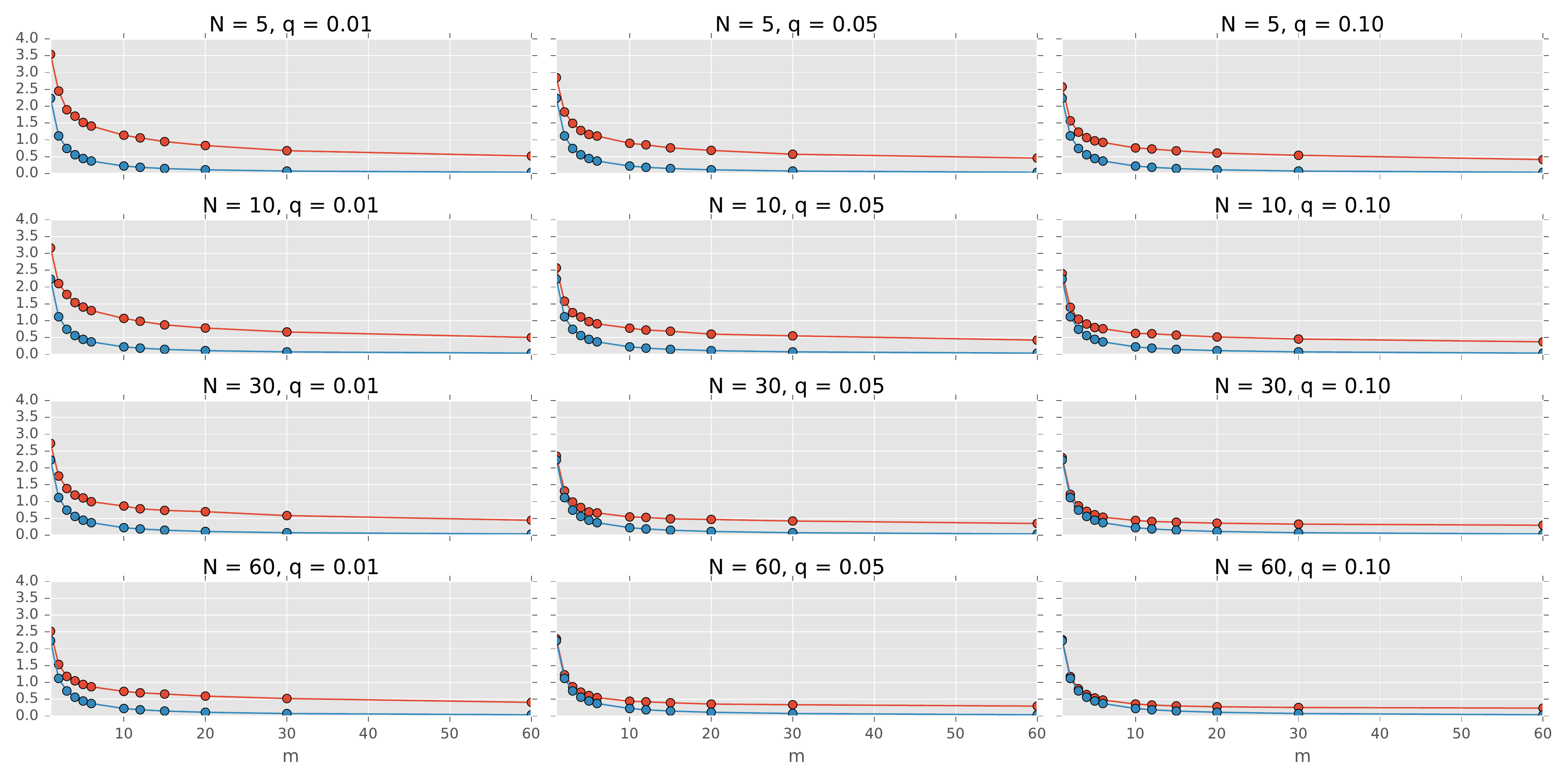}
	\end{figure}

	\renewcommand{\theequation}{A\arabic{equation}}
	\renewcommand{\theproposition}{A\arabic{proposition}}
	\renewcommand{\thetheorem}{A\arabic{theorem}}
	\renewcommand{\thecorollary}{A\arabic{corollary}}
	\setcounter{equation}{0}
		
	\section*{Appendix}
	
	All random variables in the statement below are supposed to be defined on some probability space $(\Omega,\,\mathcal{F},\,\P)$, where $\Omega$ is an arbitrary nonempty set.
	
	Our notation is mostly standard, in particular, the symbol $\I$ is used to denote the indicator function of an event
	and $\N_0$ is the set of nonnegative integers.
	
	\begin{proposition} \label{prop:E}
		Let $n_t\colon\Omega\mapsto\N_0$, $\xi_{t,\,i}\colon\Omega\mapsto[0,\,1]$, $d_{t,\,i}\colon\Omega\mapsto[0,\,+\infty)$, and $a_{t,\,i}\colon\Omega\mapsto\{0,\,1\}$, $t\in\mathbb{Z}$, $i=1,\,2,\,\ldots$, be independent collections of random variables.
		Assume that both $\xi_{t,\,i}$, and $d_{t,\,i}$, $t\in\mathbb{Z}$, $i=1,\,2,\,\ldots$, are identically distributed with distribution functions $G$ and $F$, respectively. Moreover, suppose that $n_t$, $t\in\mathbb{Z}$, are integrable, and that for each $i=0,\,1,\,\ldots$ and $t\in\mathbb{Z}$ it holds that $\P\{a_{t,\,i}\} = q$, $\P\{\xi_{t,\,i}\in (0,\,1)\} = 1$, and $\P\{d_{t,\,i} \in [0,\,m\ell]\} = 1$, where $\ell > 0$ and $m\in\N_0$. Let 
		$$
		k_t = \sum_{j=-\infty}^\infty \sum_{i=1}^{n_{t-j}} a_{t-j,\,i} \I\{(t-1)\ell \le (t-j-1 +\xi_{t-j,\,i})\ell + d_{t-j,\,i}\le t\ell\},\quad  t\in\mathbb{Z}.
		$$

		Then for all $t\in\mathbb{Z}$ it hols that
		\begin{itemize}
			\item[\rm 1)] $k_t$ is almost surely finite,
			\item[\rm 2)] $k_t$ is integrable,
			\item[\rm 3)] the equality \eqref{E=} is true for $p_j$, $j=0,\,1,\,\ldots,\,m$, defined by \eqref{p=} with
			$$
			I_j = \int F((j+1-\xi)\ell) dG(\xi),\quad j=0,\,1,\,\ldots,\,m.
			$$
		\end{itemize}
		If, in addition,
		\begin{equation} \label{G=}
		G(\xi) = \begin{cases}
		0,\quad \xi < 0,\\
		\xi,\quad \xi\in [0,\,1],\\
		1,\quad \xi > 1,
		\end{cases}
		\end{equation}
		then
		$I_j = \frac 1 \ell\int _{j\ell}^{(j+1)\ell} F(\tau)d\tau$ for each $j=0,\,1,\,\ldots,\,m$.
	\end{proposition}
	\begin{proof}
		Since the random variables $d_{t,\,i}$ and $\xi_{t,\,i}$, $t\in\mathbb{Z}$, $i=1,\,2,\,\ldots$, almost surely lie in the intervals $[0,\,m\ell]$ and $(0,\,1)$, respectively, it holds that 
		$$
		\P\{(t-1)\ell \le (t-j-1 + \xi_{\tau,\,i})\ell + d_{\tau,\,i} \le t\ell\} = 0
		$$
		when $j < 0$ or $j > m$ for all $i=1,\,2,\,\ldots$ and all $t\in\mathbb{Z}$. Therefore
		$$
		k_t \pneq \sum_{j=0}^{m} \sum_{i=1}^{n_{t-j}} a_{t-j,\,i} \I\{(t-1)\ell \le (t-j-1 +\xi_{t-j,\,i})\ell + d_{t-j,\,i}\le t\ell\}\quad\forall\,t\in\mathbb{Z}
		$$
		from which 1) and 2) readily follow. Moreover, under the stated independence and distributional assumptions we have
		\begin{eqnarray*}
			\E[k_t\mid n_t,\,n_{t-1},\,\ldots,\,n_{t-m}] = \sum_{j=0}^m q \P\{(t-1)\ell \le (t-j-1 +\xi_{1,\,1})\ell + d_{1,\,1}\le t\ell\} n_{t-j}
		\end{eqnarray*}
		for all $t\in\mathbb{Z}$. In order to obtain 3) it remains to compute $\P\{(t-1)\ell \le (t-j-1 +\xi_{1,\,1})\ell + d_{1,\,1}\le t\ell\}$ for $j=0,\,1,\,\ldots,\,m$ by applying Fubini's theorem,
		\begin{eqnarray*}
			\P\{(t-1)\ell \le (t-j-1 +\xi_{1,\,1})\ell + d_{1,\,1}\le t\ell\} &=&
			\int \left(F((j+1-\xi)\ell) - F((j - \xi)\ell)\right) dG(\xi),
		\end{eqnarray*}
		and note that $\int F(-\xi\ell)dG(\xi) = 0$ because $F(\tau) = 0$ for $\tau < 0$ and $\xi\in(0,\,1)$ with probability one. This completes the proof as the last assertion of the proposition is trivial.
	\end{proof}

\end{document}